\ifx\SICOMPVerFlag\undefined
\documentclass[11pt]{article}
\else
\documentclass[review,onefignum,onetabnum]{SICOMP/siamart171218} \fi

\IfFileExists{sariel_computer.sty}{\def\sarielComp{1}}{}
\ifx\sarielComp\undefined%
\newcommand{\SarielComp}[1]{}
\newcommand{\NotSarielComp}[1]{#1}%
\else
\newcommand{\SarielComp}[1]{#1}%
\newcommand{\NotSarielComp}[1]{}%
\fi
\newcommand{\IfPrinterVer}[2]{#2}%

\newcommand{\UsePackage}[1]{%
  \IfFileExists{../styles/#1.sty}{%
      \usepackage{../styles/#1}%
   }{%
      \IfFileExists{./styles/#1.sty}{%
         \usepackage{styles/#1}%
      }{%
         \usepackage{#1}%
      }%
   }%
}

\ifx\SICOMPVerFlag\undefined
\newcommand{\SICOMPVer}[1]{}
\newcommand{\NotSICOMPVer}[1]{#1}
\else
\newcommand{\SICOMPVer}[1]{#1}
\newcommand{\NotSICOMPVer}[1]{}
\fi

\NotSICOMPVer{
\usepackage[cm]{fullpage}%
}

\usepackage{amsmath}%
\usepackage{amssymb}%
\usepackage{bm}%
\usepackage{bbm}
\usepackage{xcolor}%
\usepackage{upgreek}%
\usepackage{euscript}%
\NotSICOMPVer{\usepackage{caption}}%
\UsePackage{picins}%
\usepackage{stmaryrd}%

\SarielComp{\usepackage{sariel_colors}}%

\usepackage[amsmath,thmmarks]{ntheorem}%
\theoremseparator{.}%

\NotSICOMPVer{
\usepackage{titlesec}%
\titlelabel{\thetitle. }%

\newcommand{\periodafter}[1]{#1.}
\titleformat{\paragraph}[runin]
  {\normalfont\bfseries}
  {\theparagraph}
  {1em}
  {\periodafter}
}

\usepackage{xcolor}%
\usepackage{mleftright}%
\usepackage{xspace}%
\usepackage{scalerel}%
\usepackage{ifluatex}
\usepackage{ifxetex}

\ifluatex %
      \usepackage{fontspec}
      \usepackage[utf8]{luainputenc}
          \usepackage[english]{babel}
          \usepackage{csquotes}
\else
       \ifxetex %
          \usepackage{fontspec}
       \else
          \usepackage[T1]{fontenc}
          \usepackage[utf8]{inputenc}
          \usepackage[english]{babel}
          \usepackage{csquotes}
       \fi
\fi
\usepackage[inline]{enumitem}

\newlist{compactenumA}{enumerate}{5}%
\setlist[compactenumA]{topsep=0pt,itemsep=-1ex,partopsep=1ex,parsep=1ex,%
   label=(\Alph*)}%
\SICOMPVer{%
   \setlist[compactenumA]{topsep=5pt,itemsep=0ex,partopsep=1ex,%
      parsep=1ex,%
      leftmargin=5ex,
      label=(\Alph*)}%
}
\newlist{compactenumi}{enumerate}{5}%
\setlist[compactenumi]{topsep=0pt,itemsep=-1ex,partopsep=1ex,parsep=1ex,%
   label=(\roman*)}%

\newlist{paraenumi}{enumerate*}{5}%
\setlist[paraenumi]{topsep=0pt,itemsep=-1ex,partopsep=1ex,parsep=1ex,%
   label=(\roman*)}%

\usepackage{ifluatex}
\usepackage{ifxetex}

\ifluatex %
      \usepackage{fontspec}
      \usepackage[utf8]{luainputenc}
\else
       \ifxetex %
          \usepackage{fontspec}
       \else
          \usepackage[T1]{fontenc}
          \usepackage[utf8]{inputenc}
       \fi
\fi
\providecommand{\BibLatexMode}[1]{}
\providecommand{\BibTexMode}[1]{#1}

\ifx\UseBibLatex\undefined%
  \renewcommand{\BibLatexMode}[1]{}
  \renewcommand{\BibTexMode}[1]{#1}
\else
  \renewcommand{\BibLatexMode}[1]{#1}
  \renewcommand{\BibTexMode}[1]{}
\fi

\BibLatexMode{%
   \usepackage[bibencoding=utf8,style=alphabetic,backend=biber,
   sortlocale=en_US]{biblatex}%
   \usepackage[english]{babel}
   \usepackage{csquotes}
   \UsePackage{sariel_biblatex}%
}

\newcommand{\hrefb}[3][black]{\href{#2}{\color{#1}{#3}}}%

\IfPrinterVer{%
   \usepackage{hyperref}%
}{%
   \usepackage{hyperref}%
   \hypersetup{%
      breaklinks,%
      ocgcolorlinks, colorlinks=true,%
      urlcolor=[rgb]{0.25,0.0,0.0},%
      linkcolor=[rgb]{0.5,0.0,0.0},%
      citecolor=[rgb]{0,0.2,0.445},%
      filecolor=[rgb]{0,0,0.4},
      anchorcolor=[rgb]={0.0,0.1,0.2}%
   }
}

\definecolor{blue25}{rgb}{0,0,0.7}
\providecommand{\emphic}[2]{%
   \textcolor{blue25}{%
      \textbf{\emph{#1}}}%
   \index{#2}}
\providecommand{\emphi}[1]{\emphic{#1}{#1}}

\theoremseparator{.}%

\theoremstyle{plain}%

\NotSICOMPVer{
\newtheorem{theorem}{Theorem}[section]
\newtheorem{lemma}[theorem]{Lemma}

\newtheorem{corollary}[theorem]{Corollary}

}

\SICOMPVer{
\newtheorem{defn}[theorem]{Definition}
\newtheorem*{defn:unnumbered}[FakeCounter]{Definition}
}

\newtheorem{observation}[theorem]{Observation}

\theoremstyle{plain}%
\theoremheaderfont{\sf} \theorembodyfont{\upshape}%
\newtheorem*{remark:unnumbered}[FakeCounter]{Remark}%

\NotSICOMPVer{

\newtheorem{defn}[theorem]{Definition}
\newtheorem*{defn:unnumbered}[FakeCounter]{Definition}
}

\NotSICOMPVer{
\newcommand{\myqedsymbol}{\rule{2mm}{2mm}}

\theoremheaderfont{\em}%
\theorembodyfont{\upshape}%
\theoremstyle{nonumberplain}%
\theoremseparator{}%
\theoremsymbol{\myqedsymbol}%
\newtheorem{proof}{Proof:}%

}

\newcommand{\atgen}{\symbol{'100}}
\newcommand{\SarielThanks}[1]{\thanks{Department of Computer Science;
      University of Illinois; 201 N. Goodwin Avenue; Urbana, IL,
      61801, USA; {\tt sariel\atgen{}illinois.edu}; {\tt
         \url{http://sarielhp.org/}.} #1}}

\newcommand{\MitchellThanks}[1]{%
   \thanks{%
      Department of Computer Science;
      University of Illinois; 201 N. Goodwin Avenue; Urbana, IL,
      61801, USA; {\tt mfjones2\atgen{}illinois.edu}; {\tt
         \url{http://mfjones2.web.engr.illinois.edu/}.} #1}}

\numberwithin{figure}{section}%
\numberwithin{table}{section}%
\numberwithin{equation}{section}%

\newcommand{\HLink}[2]{\hyperref[#2]{#1~\ref*{#2}}}
\newcommand{\HLinkSuffix}[3]{\hyperref[#2]{#1\ref*{#2}{#3}}}

\newcommand{\figlab}[1]{\label{fig:#1}}
\newcommand{\figref}[1]{\HLink{Figure}{fig:#1}}

\newcommand{\thmlab}[1]{{\label{theo:#1}}}
\newcommand{\thmref}[1]{\HLink{Theorem}{theo:#1}}

\newcommand{\corlab}[1]{\label{cor:#1}}

\newcommand{\seclab}[1]{\label{sec:#1}}
\newcommand{\secref}[1]{\HLink{Section}{sec:#1}}

\providecommand{\deflab}[1]{\label{def:#1}}
\newcommand{\defref}[1]{\HLink{Definition}{def:#1}}

\newcommand{\apndlab}[1]{\label{apnd:#1}}
\newcommand{\apndref}[1]{\HLink{Appendix}{apnd:#1}}

\newcommand{\tbllab}[1]{\label{table:#1}}
\newcommand{\tblref}[1]{\HLink{Table}{table:#1}}

\newcommand{\lemlab}[1]{\label{lemma:#1}}
\newcommand{\lemref}[1]{\HLink{Lemma}{lemma:#1}}%

\newcommand{\itemlab}[1]{\label{item:#1}}
\newcommand{\itemref}[1]{\HLinkSuffix{}{item:#1}{}}

\newcommand{\obslab}[1]{\label{obsv:#1}}
\newcommand{\obsref}[1]{\HLink{Observation}{obsv:#1}}

\providecommand{\eqlab}[1]{}%
\renewcommand{\eqlab}[1]{\label{equation:#1}}

\providecommand{\remove}[1]{}%
\newcommand{\Set}[2]{\left\{ #1 \;\middle\vert\; #2 \right\}}
\newcommand{\pth}[2][\!]{\mleft({#2}\mright)}%
\newcommand{\pbrcx}[1]{\left[ {#1} \right]}%

\newcommand{\ceil}[1]{\left\lceil {#1} \right\rceil}
\newcommand{\floor}[1]{\left\lfloor {#1} \right\rfloor}

\newcommand{\brc}[1]{\left\{ {#1} \right\}}
\newcommand{\cardin}[1]{\left| {#1} \right|}%

\renewcommand{\th}{th\xspace}

\renewcommand{\Re}{\mathbb{R}}%
\newcommand{\Na}{\mathbb{N}}

\DefineNamedColor{named}{OliveGreen} {cmyk}{0.44,0,0.95,0.90}

\NotSICOMPVer{
\setcounter{secnumdepth}{5}
}

\newcommand{\ts}{\hspace{0.6pt}}
\newcommand{\Zorder}{\ensuremath{\mathrm{Z}}-order\xspace}%
\newcommand{\URorder}{%
    \protect\raisebox{-0.5pt}{%
      \protect\rotatebox{90}{\ensuremath{\mathrm{U}}}%
    }%
    \ts-order\xspace}%
\newcommand{\ZBorder}{%
    \protect\raisebox{-0.5pt}{%
      \protect\reflectbox{\ensuremath{\mathrm{Z}}}%
    }%
    \ts-order\xspace}%

\providecommand{\Mh}[1]{#1}

\newcommand{\eps}{\varepsilon}
\newcommand{\epsA}{\eps}%
\newcommand{\lgeps}{\lg(1/\eps)}
\newcommand{\logeps}{\log(1/\eps)}
\newcommand{\Oe}{O_\eps}
\newcommand{\Od}{O_d}
\newcommand{\Ode}{O_{d,\eps}}

\newcommand{\dist}[1]{\| #1 \|}
\newcommand{\diamX}[1]{\mathsf{diam}\pth{#1}}
\newcommand{\abs}[1]{\left\vert #1 \right\vert}

\newcommand{\zeroone}{\brc{0,1}}

\newcommand{\pset}{P}
\newcommand{\rset}{R}
\newcommand{\bset}{B}

\newcommand{\pt}{p}%
\newcommand{\ptq}{q}%
\newcommand{\pts}{s}
\newcommand{\ptt}{t}%
\newcommand{\ptr}{\Mh{r}}%
\newcommand{\ptb}{\Mh{b}}%

\newcommand{\ptopt}{\pt^\star}

\newcommand{\qt}{\mathcal{T}}
\newcommand{\qte}{\mathcal{T}_\eps}

\newcommand{\grid}{\mathsf{G}}
\newcommand{\cell}{\Box}%

\newcommand{\order}{\sigma}
\newcommand{\orderset}{\Pi}
\newcommand{\ordAll}{\orderset^+}%
\newcommand{\OrdChar}{\mathfrak{O}}%
\newcommand{\OrdY}[2]{\mathfrak{O}\pth{#1,#2}}%

\newcommand{\xor}{\oplus}
\newcommand{\msb}{\mathrm{msb}}
\newcommand{\indic}[1]{\mathbbm{1}\!\pbrcx{#1}}

\newcommand{\graph}{G}
\newcommand{\graphdist}[2]{\mathsf{d}_{\graph}\pth{#1, #2}}
\newcommand{\graphedges}{E}

\newcommand{\distY}[2]{\left\| {#1} - {#2} \right\|}
\newcommand{\distSetY}[2]{d\pth{ #1, #2}}%

\newcommand{\SaveContent}[2]{%
   \expandafter\newcommand{#1}{#2}%
}

\newcommand{\RestatementOf}[2]{
   \noindent%
   \textbf{Restatement of #1.}
   {\em #2{}}%
}

\newcommand{\IntRange}[1]{\mleft\llbracket #1 \mright\rrbracket}
\newcommand{\IRX}[1]{\IntRange{#1}}%

\newcommand{\Term}[1]{\textsf{#1}}

\newcommand{\pA}{\Mh{p}}%
\newcommand{\pB}{\Mh{q}}%
\newcommand{\pC}{\Mh{u}}%

\newcommand{\PS}{\Mh{P}}%
\newcommand{\AccSet}{\mathcal{Z}}

\newcommand{\DFS}{\Term{DFS}\xspace}%
\newcommand{\ANN}{\Term{ANN}\xspace}%
\newcommand{\MST}{\Term{MST}\xspace}%
\newcommand{\LSH}{\Term{LSH}\xspace}%

\newcommand{\Tree}{\Mh{T}}%
\newcommand{\lso}{locality-sensitive ordering\xspace}%
\newcommand{\Lso}{Locality-sensitive ordering\xspace}%
\newcommand{\VFTS}{\Term{VFTS}\xspace}

\newcommand{\MetricSpace}{\mathcal{M}}
\newcommand{\MetricDist}{d}

\newcommand{\ZZ}{{\mathbb{Z}}}%
\newcommand{\Cube}{\Mh{\mathcal{C}}}%
\newcommand{\sidelengthX}[1]{\mathsf{sidelength}\pth{#1}}%
\newcommand{\Bmod}{\mathop{\raisebox{0.2ex}{\scalebox{.7}{\%}}}}%
\newcommand{\etal}{et~al.\xspace}

\allowdisplaybreaks[1]

\newcommand{\si}[1]{#1}

\providecommand{\TPDF}[2]{\texorpdfstring{#1}{#2}}

\SICOMPVer{

\headers{On Locality-Sensitive Orderings}{
T. M. Chan, S. Har-Peled, and M. Jones}

\title{On Locality-Sensitive Orderings and Their Applications%
   \thanks{Submitted to the editors 02/22/19. A preliminary version of
      this paper appeared in I{TC}S 2019 \cite{chj-lso-19}.}}

\author{Timothy M. Chan%
   \thanks{Department of Computer Science, University of Illinois at
      Urbana-Champaign, Urbana, IL 61801 (\email{tmc@illinois.edu}).
      \funding{Work on this paper was partially supported by NSF AF
         award CCF-1814026.}}%
   \and%
   Sariel Har-Peled%
   \thanks{Department of Computer Science, University of Illinois at
      Urbana-Champaign, Urbana, IL 61801
      (\email{sariel@illinois.edu}). %
      \funding{Work on this paper was partially supported by NSF AF
         awards CCF-1421231 and CCF-1907400.}}%
   \and%
   Mitchell Jones%
   \thanks{Department of Computer Science, University of Illinois at
      Urbana-Champaign, Urbana, IL 61801
      (\email{mfjones2@illinois.edu}).}}

\ifpdf
\hypersetup{
  pdftitle={On Locality-Sensitive Orderings and Their Applications},
  pdfauthor={T. M. Chan, S. Har-Peled, and M. Jones}
}
\fi

\externaldocument{ex_supplement}

}

\NotSICOMPVer{
\BibLatexMode{%
}

\title{On Locality-Sensitive Orderings and Their Applications%
   \thanks{A preliminary version of this paper appeared in I{TC}S 2019
      \cite{chj-lso-19}.}%
}%
\date{\today}

\author{%
   Timothy M. Chan%
   \thanks{%
      Department of Computer Science, University of Illinois at
      Urbana-Champaign, {\tt \{t{m}c}.  Work on this paper was
      partially supported by NSF AF award CCF-1814026.%
   }%
   \and%
   Sariel Har-Peled%
   \SarielThanks{Work on this paper was partially supported by NSF
      AF awards CCF-1421231 and CCF-1907400.}%
   \and%
   Mitchell Jones%
   \MitchellThanks{%
      Work on this paper was partially supported by NSF AF awards
      CCF-1421231 and CCF-1907400.%
   }%
}%
}

\begin{document}

\maketitle

\begin{abstract}
    For any constant $d$ and parameter $\eps \in (0,1/2]$, we show the
    existence of (roughly) $1/\eps^d$ orderings on the unit cube
    $[0,1)^d$, such that for any two points $\pt,\ptq\in [0,1)^d$
    close together under the Euclidean metric, there is a linear
    ordering in which all points between $\pt$ and $\ptq$ in the
    ordering are ``close'' to $\pt$ or $\ptq$.  More precisely, the
    only points that could lie between $\pt$ and $\ptq$ in the
    ordering are points with Euclidean distance at most
    $\eps\distY{\pt}{\ptq}$ from either $\pt$ or $\ptq$. These
    orderings are extensions of the \Zorder, and they can be
    efficiently computed.

    Functionally, the orderings can be thought of as a replacement to
    quadtrees and related structures (like well-separated pair
    decompositions). We use such orderings to obtain surprisingly
    simple algorithms for a number of basic problems in
    low-dimensional computational geometry, including
    \begin{enumerate*}[label=(\roman*)]
        \item dynamic approximate bichromatic closest pair,
        \item dynamic spanners,
        \item dynamic approximate minimum spanning trees,
        \item static and dynamic fault-tolerant spanners, and
        \item approximate nearest neighbor search.
    \end{enumerate*}
\end{abstract}

\SICOMPVer{

}

\section{Preface}

In this paper, we describe a technique that leads to new, simpler
algorithms for a number of fundamental proximity problems in
low-dimensional Euclidean spaces.

Given data, having an ordering over it is quite useful---it enables
one to sort it, store it, and search it efficiently, among other
things. Such an order is less natural for points in the plane (or in
higher dimensions). One way to impose such orders is by using
bijective mappings from the plane to the line (which has a natural
order, and thus endows the plane with an order). Such mappings, known
as space-filling curves, were discovered in 1890 by Peano
\cite{p-sucqr-1890}. (See also the book by Sagan \cite{s-sfc-94} for
more information on space-filling curves.) For computational purposes,
the \Zorder, a somewhat inferior space-filling curve, is the easiest
to implement as it is easily computed by interleaving the bits of the
$x$ and $y$ coordinates.

A natural property one desires in an ordering of the plane is that it
preserves locality---points that are close together geometrically
remain close in the resulting ordering. Unfortunately, no
mapping/ordering can have this property universally, as the topology
of the line and the plane are fundamentally different. Nevertheless,
the \Zorder already has some nice locality properties---it maps certain
squares to intervals on the real line, and these squares forms grids
that cover the unit square. Furthermore, these grids are universal, in
the sense that there is a grid for any desired resolution.

To get better locality properties, one has to use more orders. It is
known that if one uses three orders in the plane (which is the result
of shifting the plane before applying the \Zorder), then for any axis
parallel square $\Cube$ (inside the unit square), there exists a
square $\Cube'$ that contains $\Cube$, such that $\Cube'$ is only
slightly bigger than $\Cube$, and one of the three orders maps
$\Cube'$ to an interval.

Our purpose here is to get an even stronger locality property, which
requires a larger collection of orderings. Specifically,
consider two points $\pt, \pt' \in [0,1]^2$. The desired property is
that there are two squares $\Cube$ and $\Cube'$, and an order
$\order$ in the collection, with the following properties:
\begin{paraenumi}
    \item $\pt \in \Cube$ and $\pt' \in \Cube'$,
    \item the diameters of $\Cube$ and $\Cube'$ are only an
    $\eps$-fraction of the distance between $\pt$ and $\pt'$,
    \item $\Cube$ and $\Cube'$ are mapped to two intervals on the real
    line by $\order$, and
    \item these two intervals are adjacent.
\end{paraenumi}
Such an ordering $\order$ with the desired properties is illustrated
in \figref{orderly}.
\begin{figure}[t]
    \centerline{%
       \includegraphics[scale=0.8]{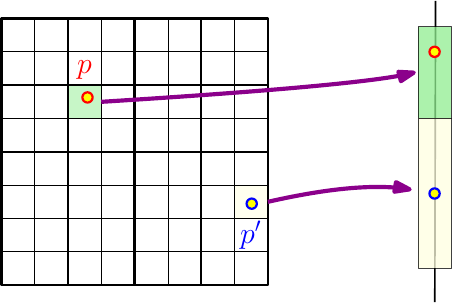}%
    }%
    \caption{}
    \figlab{orderly}
\end{figure}

For algorithmic applications, this collection of orders needs to be
small, and it needs to be easily computable. Surprisingly, we show
that the desired collection of orders has size that depends only on
$\eps$, and these orders can be easily computed.

To see why having such a collection of orders is so useful, consider
the problem of computing the closest pair of points in a given set of
points $\PS$. Every order in the collection induces an ordering of
$\PS$. Furthermore, the closest pair of points are going to be
adjacent in one of these orders, and as such can be readily computed
by considering all consecutive pairs of points in the ordering (the
number of such pairs is linear). Furthermore, using balanced binary
search trees, it is easy to maintain each ordered set under insertions
and deletions.  Therefore, one can maintain the closest pair of points
by storing $\PS$ in such a data structure for each of the orderings.
As a result, a dynamic problem that might seem in advance somewhat
challenging reduces (essentially) to the mundane task of maintaining
ordered sets under insertions and deletions.

\NotSICOMPVer{ \bigskip%
   \vspace{-2\lineskip}%
}

\begin{figure}
    \includegraphics[scale=0.5,page=1]{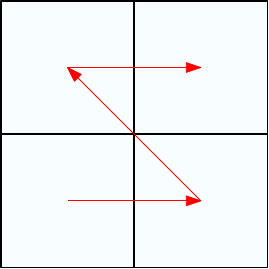}%
    \hfill%
    \includegraphics[scale=0.5,page=2]{figs/qdfs_2}%
    \hfill%
    \includegraphics[scale=0.5,page=3]{figs/qdfs_2}%
    \hfill%
    \includegraphics[scale=0.5,page=4]{figs/qdfs_2}%
    \hfill%
    \includegraphics[scale=0.5,page=5]{figs/qdfs_2}%

    \smallskip%
    
    \includegraphics[scale=0.5,page=1]{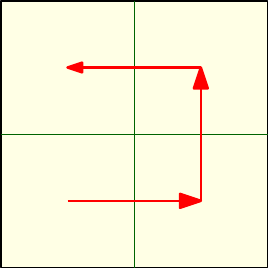}%
    \hfill%
    \includegraphics[scale=0.5,page=2]{figs/qdfs_3}%
    \hfill%
    \includegraphics[scale=0.5,page=3]{figs/qdfs_3}%
    \hfill%
    \includegraphics[scale=0.5,page=4]{figs/qdfs_3}%
    \hfill%
    \includegraphics[scale=0.5,page=5]{figs/qdfs_3}%
    \caption{Changing the order in which a \DFS visits the children of
    a quadtree node induces a different ordering of the underlying
    square (and produces different space filling curves). The top row
    shows the \Zorder (or \ZBorder), and the bottom row shows the
    {\protect\URorder{}}.}
    \figlab{zig:zag}
\end{figure}

\section{Introduction}

\paragraph*{Quadtrees and \Zorder} %
Consider a point set $\PS \subseteq [0,1)^2$, its quadtree, and a 
depth-first search (\DFS) traversal of this quadtree. One can order
the points of $\PS$ according to this traversal, resulting in some
ordering $\prec$ of the underlying set $[0,1)^2$. The relation 
$\prec$ is the ordering along some space filling mapping.

One particular ordering of interest is the \emphi{\Zorder{}}.
Conceptually speaking, the \Zorder{} can be thought of as a \DFS of
the quadtree over $[0,1)^2$, where the children of each node in the
quadtree are always visited in the same p{r}e-defined order (see
\figref{zig:zag}). The \Zorder is a total ordering over the points in
$[0,1)^2$, and can be formally defined by a bijection $z$ from the
unit interval $[0,1)$ to the unit square $[0,1)^2$. Given a real
number $\alpha \in [0,1)$, with the binary expansion
$\alpha=0.x_1 x_2 x_3 \ldots$ (i.e.,
$\alpha = \sum_{i=1}^\infty x_i 2^{-i}$), the \Zorder mapping of
$\alpha$ is the point
$z(\alpha) = (0.x_2x_4x_6\ldots, 0.x_1x_3x_5\ldots )$. 
We note that the \Zorder mapping $z$ is not continuous. Nevertheless,
the \Zorder mapping has the advantage of being easy to define.  In
particular, computing the \Zorder or its inverse is quite easy, if one
is allowed bit{}wise-logical operations---in particular, the ability to
compute compressed quadtrees efficiently is possible only if such
operations are available \cite{h-gaa-11}. The approach extends to 
higher constant dimensions. 

The idea of using the \Zorder can be traced back to the work of Morton
\cite{m-c-66}, and it is widely used in databases and seems to improve
performance in practice \cite{kf-oprt-93}.  Once comparison by \Zorder
is available, building a compressed quadtree is no more than storing
the points according to the \Zorder, and this yields simple data
structures for various problems. For example, Liao
\etal~\cite{lll-hdsssfc-01} and Chan
\cite{c-cppsr-02,c-miann-06,c-wspdl-08} applied the \Zorder to obtain
simple efficient algorithms for approximate nearest neighbor search
and related problems.

\paragraph*{Shifting}
The \Zorder (and quadtrees) does not preserve distance. That is, two
points that are far away might be mapped to two close-together points,
and vice versa. This problem is even apparent when using a grid, where
points that are close together get separated into different grid
cells. One way to get around this problem is to shift the grid
(deterministically or randomly) \cite{hm-ascpp-85}. The same approach
works for quadtrees---one can shift the quadtree constructed for a
point set several times such that for any pair of points in the
quadtree, there will be a shift where the two points are in a cell of
diameter that is $\Od(1)$ times their distance. (Throughout,
we use the $\Od$ notation to hide constants that depend on $d$.
Similarly, $\Oe$ hides dependencies on $\eps$.) Improving an earlier
work by Bern \cite{b-acpqh-93}, Chan \cite{c-annqr-98} showed that
$2\lceil d/2\rceil+1$ deterministic shifts are enough in $d$
dimensions (a proof is reproduced in \apndref{shifting}).  A somewhat
similar shifting scheme was also suggested by Feige and Krauthgamer
\cite{fk-sfpta-97}.  Random shifting of quadtrees underlines, for
example, the approximation algorithm by Arora for Euclidean TSP
\cite{a-ptase-98}.

By combining \Zorder with shifting, both Chan \cite{c-cppsr-02}
and Liao \etal~\cite{lll-hdsssfc-01} observed an
extremely simple data structure for $\Od(1)$-\si{appr}\-\si{oximate} 
nearest neighbor search in constant dimensions: just store the points in
\Zorder for each of the $2\lceil d/2\rceil+1$ shifts; given a query
point $q$, find the successor and predecessor of $q$ in the \Zorder by
binary search for each of the shifts, and return the closest point
found.  The data structure can be easily made dynamic to support
insertions and deletions of points, and can also be adapted to find
$\Od(1)$-approximate bichromatic closest pairs.

For approximate nearest neighbor (\ANN) search, the $\Od(1)$
approximation factor can be reduced to $1+\eps$ for any fixed
$\eps>0$, though the query algorithm becomes more involved
\cite{c-cppsr-02} and unfortunately cannot be adapted to compute
$(1+\eps)$-approximate bichromatic closest pairs dynamically.  (In the
monochromatic case, however, the approach can be adapted to find
\emph{exact} closest pairs, by considering $O_d(1)$ successors and
predecessors of each point~\cite{c-cppsr-02}.)

For other proximity-related problems such as spanners and approximate
minimum spanning trees (\MST), this approach does not seem to work as
well: for example, the static algorithms in \cite{c-wspdl-08}, which
use the \Zorder, still requires explicit constructions of compressed
quadtrees and are not easily dynamizable.

\paragraph*{Main new technique: \Lso{}s}
For any given $\epsA > 0$, we show that there is a family of
$O_d((1/\epsA^{d}) \log (1/\epsA))$ orderings of $[0,1)^d$ with the
following property: For any $\pt, \ptq \in [0,1)^d$, there is an
ordering in the family such that all points lying between $\pt$ and
$\ptq$ in this ordering are within distance at most
$\epsA \distY{\pt}{\ptq}$ from either $\pt$ or $\ptq$ (where 
$\| \cdot \|$ is the standard Euclidean norm).
The order between two points can be determined efficiently using some
bitwise-logical operations.  See \thmref{lso}.  We refer to these as
\emphi{\lso{}s}.  They generalize the previous construction of
$2\lceil d/2\rceil+1$ shifted copies of the \Zorder, which guarantees
the stated property only for a large specific constant (equivalent to
setting $\epsA \approx d^{3/2}$).  The new refined property ensures,
for example, that a $(1+\eps)$-approximate nearest neighbor of a point
$q$ can be found among the immediate predecessors and successors of
$q$ in these orderings.
 
\paragraph*{Applications}
\Lso{}s immediately lead to simple algorithms for a number of
problems, as listed below.  Many of these results are significant
simplification of previous work; some of the results are new.
\NotSICOMPVer{\medskip}%
\begin{compactenumA}
    \NotSICOMPVer{\medskip}
    \item \textsf{Approximate bichromatic closest pair.} %
    \thmref{bichromatic:closest:pair} presents a data structure that
    maintains a $(1+\eps)$-approximate closest bichromatic pair for
    two sets of points in $\Re^d$, with an update time of $\Ode(\log
    n)$, for any fixed $\eps>0$ (the hidden factors depending on 
    $\eps$ are proportional to $(1/\eps^{d}) \log^2(1/\eps)$).
    Previously, a general technique of Eppstein \cite{e-demst-95} can
    be applied in conjunction with a dynamic data structure for \ANN,
    but the amortized update time increases by two $\log n$ factors.

    \NotSICOMPVer{\medskip}
    \item \textsf{Dynamic spanners.} %
    For a parameter $t \geq 1$ and a set of points $\pset$ in $\Re^d$,
    a graph $\graph = (\pset, \graphedges)$ is a $t$-spanner for
    $\pset$ if for all $\pt, \ptq \in \pset$, there is a $\pt$-$\ptq$
    path in $\graph$ of length at most $t \distY{\pt}{\ptq}$.  Static
    algorithms for spanners have been extensively studied in
    computational geometry. The dynamic problem appears tougher, and
    has also received much attention (see \tblref{d:y:n:spanner}).  We
    obtain a very simple data structure for maintaining dynamic
    $(1+\eps)$-spanners in Euclidean space with an update (insertion
    and deletion) time of $\Ode(\log n)$ and having $\Ode(n)$ edges in
    total, for any fixed $\eps>0$. See \thmref{dynamic:spanners}.
    Although Gottlieb and Roditty \cite{gr-odsdm-08} have previously
    obtained the same update time $\Ode(\log n)$, their
    method requires much more intricate details. (Note that Gottlieb
    and Roditty's method more generally applies to spaces with bounded
    doubling dimension, but no simpler methods have been reported in
    the Euclidean setting.)

    \NotSICOMPVer{\medskip}
    \item \textsf{Dynamic approximate minimum spanning trees.}  As is
    well-known~\cite{ck-fasgg-93,h-gaa-11}, a $(1+\eps)$-approximate
    Euclidean \MST of a point set $\pset$ can be computed from the
    \MST of a $(1+\eps)$-spanner of $\pset$.  In our dynamic spanner
    (and also Gottlieb and Roditty's method \cite{gr-odsdm-08}), each
    insertion/deletion of a point causes $\Ode(1)$ edge updates to the
    graph.  Immediately, we thus obtain a dynamic data structure for
    maintaining a $(1+\eps)$-approximate Euclidean \MST, with update
    time (ignoring dependencies on $d$ and $\eps$) equal to that for 
    the dynamic graph \MST problem, which is currently
    $O(\log^4n/\log\log n)$ with amortization~\cite{hrw-ffdms-15}.

    \NotSICOMPVer{\medskip}
    \item \textsf{Static and dynamic vertex-fault-tolerant spanners.}
    For parameters $k, t \geq 1$ and a set of points $\pset$ in
    $\Re^d$, a $k$-vertex-fault-tolerant $t$-spanner is a graph
    $\graph$ which is a $t$-spanner and for any
    $\pset' \subseteq \pset$ of size at most $k$, the graph
    $\graph \setminus \pset'$ remains a $t$-spanner for
    $\pset \setminus \pset'$.  Fault-tolerant spanners have been
    extensively studied (see \tblref{fault}).  \Lso{}s lead to a very
    simple construction for $k$-vertex-fault-tolerant
    $(1+\eps)$-spanners, with $\Ode(kn)$ edges, maximum degree 
    $\Ode(k)$, and $\Ode(n\log n + kn)$ running time.  See
    \thmref{ft:spanners}.  Although this result was known before, all
    previous constructions (including suboptimal ones), from
    Levcopoulos \etal's \cite{lns-eacft-98} to Solomon's
    work~\cite{s-fhphc-14}, as listed in \tblref{fault}, require
    intricate details. It is remarkable how effortlessly we achieve
    optimal $\Ode(k)$ degree, compared to the previous methods.
    (Note, however, that some of the more recent previous
    constructions more generally apply to spaces with bounded doubling
    dimension, and some also achieve good bounds on other parameters
    such as the total weight and the hop-diameter.)
    
    Our algorithm can be easily made dynamic, with $\Ode(\log n+k)$
    update time.  No previous results on dynamic fault-tolerant
    spanners were known.

    \NotSICOMPVer{\medskip}
    \item \textsf{Approximate nearest neighbors.} %
    \SICOMPVer{\sloppy}
    \Lso{}s lead to a simple dynamic data structure for
    $(1+\eps)$-approximate nearest neighbor search with $\Ode(\log n)$
    time per update/query.  While this result is not new
    \cite{c-cppsr-02}, we emphasize that the query algorithm is the
    simplest so far---it is just a binary search in the orderings
    maintained.
\end{compactenumA}

\begin{table}
    \centering
    \begin{tabular}{*{3}{l}}
      reference
      & insertion time
      & deletion time
      \\\hline
      Roditty \cite{r-fdgs-12}
      & $\log n$
      & $n^{1/3}\log^{O(1)}n$
        \phantom{$2^{2^{2^2}}$}\\      
      Gottlieb and Roditty \cite{gr-iafdg-08} & $\log^2n$ & $\log^3n$\\
      Gottlieb and Roditty \cite{gr-odsdm-08} & $\log n$ & $\log n$\\
      \thmref{dynamic:spanners} & $\log n$ & $\log n$
    \end{tabular}
    \caption{Previous work and our result on dynamic
    $(1+\eps)$-spanners in $\Re^d$. All bounds are of the form
    $\Ode(\,\cdot\,)$ (the hidden dependencies on $\eps$ are
    $1/\eps^{O(d)}$).}%
    \tbllab{d:y:n:spanner}
\end{table}

\begin{table}
    \centering
    \SICOMPVer{\small}
    \begin{tabular}{*{4}{l}}
      reference
      & \# edges
      & degree
      & running time
      \\
      \hline
      Levcopoulos \etal \cite{lns-eacft-98}
      & $2^{O(k)}n$
      & $2^{O(k)}$
      & $n\log n + 2^{O(k)}n$ \phantom{$2^{2^{2^2}}$}
      \\
      & $k^2 n$ & unbounded & $n\log n + k^2 n$\\
      & $kn\log n$ &  unbounded & $kn\log n$\\
      Lukovszki \cite{l-nrftg-99,l-nrgsa-99}
      & $kn$ & $k^2$ & $n\log^{d-1}n + kn\log\log n$\\
      Czumaj and Zhao \cite{cz-ftgs-04}
      & $kn$
      & $k$
      & $kn\log^d n + k^2 n\log k$\\
      H. Chan \etal \cite{clns-ndsbs-15}
      & $k^2 n$
      & $k^2$
      & $n\log n + k^2 n$\\
      Kapoor and Li~\cite{kl-ecpd-13}/Solomon \cite{s-fhphc-14} 
      & $kn$ & $k$ & $n\log n + kn$\\
      \thmref{ft:spanners}
      & $kn$ & $k$ & $n\log n + kn$
    \end{tabular}
    \caption{Previous work and our result on static
    $k$-vertex-fault-tolerant $(1+\eps)$-spanners in $\Re^d$. 
    All bounds are of the form $\Ode(\,\cdot\,)$ (the hidden 
    dependencies on $\eps$ are $1/\eps^{O(d)}$).}%
    \tbllab{fault}
\end{table}

\paragraph*{Computational models and assumptions}
The model of computation we have assumed is a unit-cost real RAM,
supporting standard arithmetic operations and comparisons (but no
floor function), augmented with bitwise-logical operations
(bitwise-exclusive-or and bitwise-and), which are commonly available
in programming languages (and in reality are cheaper than some
arithmetic operations like multiplication).

If we assume that input coordinates are integers bounded by $U$ and
instead work in the word RAM model with $(\log U)$-bit words
($U\ge n$), then our approach can actually yield \emph{sublogarithmic}
query/update time.  For example, we can achieve $\Ode(\log\log U)$
expected time for dynamic approximate bichromatic closest pair,
dynamic spanners, and dynamic \ANN, by replacing binary search with
van Emde Boas trees~\cite{e-poflt-77}.  Sublogarithmic algorithms were
known before for dynamic \ANN~\cite{c-cppsr-02}, but ours is the first
sublogarithmic result for dynamic $(1+\eps)$-spanners. 
Our results also answers the open problem of dynamic
$(1+\eps)$-approximate bichromatic closest pair in sublogarithmic
time, originally posed by Chan and Skrepetos \cite{cs-ddsah-17}. 

Throughout, we assume (without loss of generality) that
$\eps$ is a power of $2$; that is, $\eps= 2^{-E}$ for some positive
integer $E$.%

\section{\Lso{}s}

\subsection{Grids and orderings}

\begin{defn:unnumbered}
    For a set $X$, consider a total order (or ordering) $\prec$ on the
    elements of $X$. 
    Two elements $x,y \in X$ are \emphi{adjacent} if there is no
    element $z \in X$, such that $x \prec z \prec y$ or
    $y \prec z \prec x$.

    Given two elements $x,y \in X$, such that $x \prec y$, the
    \emphi{interval} $[x,y)$ is the set
    \begin{math}
        [x,y) = \brc{x} \cup \Set{z \in X}{x \prec z \prec y}.
    \end{math}
\end{defn:unnumbered}

The following is well known, and goes back to a work by Walecki in the
$19$\th century \cite{a-wwc-08}. We include a proof in
\apndref{many:orderings} for the sake of completeness.
(If we don't care about the constant factor in the number of orderings, there are other straightforward alternative proofs.)

\SaveContent{\LemmaManyOrderings}{%
   For $n$ elements $\brc{0, \ldots, n-1}$, there is a set $\OrdChar$
   of $\ceil{n/2}$ orderings of the elements, such that, for all
   $i,j \in \brc{0, \ldots, n-1}$, there exists an ordering
   $\order \in \OrdChar$ in which $i$ and $j$ are adjacent.%
}

\begin{lemma}
    \lemlab{many:orderings}%
    \LemmaManyOrderings{}
\end{lemma}

\begin{defn}
    \deflab{grids}%
    Consider an axis-parallel cube $\Cube \subseteq \Re^d$ with side
    length $\ell$. Partitioning it uniformly into a
    $t \times t\times \cdots \times t$ grid $\grid$ creates the
    \emphi{$t$-grid} of $\Cube$.  The grid $\grid$ is a set of $t^d$
    identically sized cubes with side length $\ell/t$.

    For a cube $\cell \subseteq \Re^d$, its \emphi{diameter} is
    $\diamX{\cell} = \sidelengthX{\cell} \sqrt{d}$.
\end{defn}

\begin{figure*}[t]
    \centering \hspace*{\fill}
    \includegraphics[page=1]{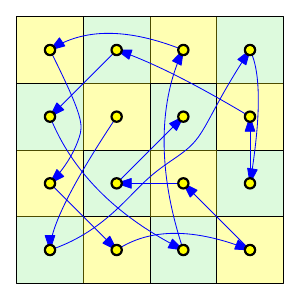} \hfill
    \includegraphics[page=2]{figs/traversal} \hspace*{\fill}
    \caption{One ordering of a set of cells.}
    \figlab{eg:ordering}
\end{figure*}

By \lemref{many:orderings} we obtain the following result.

\begin{corollary}
    \corlab{grid:orderings}%
    For a $t$-grid $\grid$ of an axis-parallel cube
    $\Cube \subseteq \Re^d$, there is a set $\OrdY{t}{d}$ of $O(t^d)$
    orderings, such that for any $\cell_1, \cell_2 \in \grid$, there
    exists an order $\order \in \OrdY{t}{d}$ where $\cell_1$ and
    $\cell_2$ are adjacent in $\order$.
\end{corollary}

\subsection{\TPDF{$\eps$}{epsilon}-Quadtrees}

\begin{defn}
    An \emphi{$\eps$-quadtree} $\qte$ is a quadtree-like structure,
    built on a cube with side length $\ell$, where each cell is 
    partitioned into a $(1/\eps)$-grid.  The construction then
    continues recursively into each grid cell of interest.  As such, a
    node in this tree has up to $1/\eps^d$ children, and a node at
    level $i \geq 0$ has an associated cube of side length
    $\ell \eps^i$. When $\eps = 1/2$, this is a regular quadtree.
\end{defn}

\begin{lemma}
    \lemlab{reg:to:eps:quad}%
    Let $E > 0$ be an integer, $\eps = 2^{-E}$, and $\qt$ be a
    regular quadtree over $[0, 2)^d$.  There are $\eps$-quadtrees
    $\qte^0, \ldots, \qte^{E-1}$, such that the collection of cells 
    at each level in $\qt$ is contained in exactly one of
    these $\eps$-quadtrees.
\end{lemma}
\begin{proof}
    \SICOMPVer{\sloppy}
    
    For $i = 0, \ldots, E-1$, construct the $\eps$-quadtree $\qte^i$
    using the cube $\bigl[0,2^{E-i+1}\bigr)^d \supseteq [0,2)^d$ as the
    root. Now for $j \in \brc{0, \ldots, E-1}$, observe that the 
    collection of cells at levels $j, j + E, j + 2E, \ldots,$ of $\qt$
    will also be in the quadtree $\qte^{j}$. 
    Indeed, any node at level $j + \ell E$ in $\qt$
    corresponds to a cell of side length $2^{-(j + \ell E)+1}$. Now in
    the $(\ell+1)$\th level of quadtree $\qte^j$, this same node will 
    have side length $\eps^{\ell+1} 2^{E-j+1} = 2^{-(j + \ell E)+1}$.
\end{proof}

Consider an $\eps$-quadtree $\qte$. Every node has up to $1/\eps^d$
children. Consider any ordering $\order$ of 
$\brc{1, \ldots, 1/\eps^d}$. Conceptually speaking, this induces a 
\DFS of $\qte$ that always visits the children of a node in the order
specified by $\order$. This induces an
ordering on the points in the cube which is the root of
$\qte$. Indeed, for any two points, imagine storing them in an
$\eps$-quadtree---this implies that the two points are each stored in
their own leaf node, which contains no other point of interest.  Now,
independently of what other points are stored in the quadtree, this
\DFS traversal would visit these two points in the same order. This
can be viewed as a space filling curve (which is not continuous) which
maps a cube to an interval. This is a generalization of the \Zorder.
In particular, given a point set stored in $\qte$, and given $\order$,
one can conceptually order the points according to this \DFS
traversal, resulting in $1$-dimensional ordering of the points. We
denote the resulting ordering by $(\qte, \order)$. 

In \secref{compare}, we show that given $(\qte, \order)$, the order of
any two points in $[0,2)^d$ can be determined efficiently, and avoids
explicitly handling this \DFS traversal of $\qte$. Alternatively, the
\DFS on $\qte$ (according to $\order$) is implicitly defined by the
total ordering $(\qte, \order)$ of points in $[0,2)^d$.

\begin{defn}
    \deflab{o:set}%
    Let $\orderset$ be the set of all orderings of $[0,2)^d$, induced
    by picking one of the $\lg(1/\eps)$ trees of
    \lemref{reg:to:eps:quad}, together with an ordering
    $\order \in \OrdY{1/\eps}{d}$, as defined by
    \lemref{many:orderings}.  Each ordering in $\orderset$ is an
    \emphi{$\eps$-ordering}.
\end{defn}

Suppose there are two points which lie in a quadtree cell
that has diameter close to their distance. Formally, consider two
points $\pA, \pB \in [0,1)^d$, a parameter $\eps > 0$, such that 
$\pA, \pB$ are both contained in a cell $\cell$ of the regular
quadtree $\qt$ with $\diamX{\cell} \leq 2\distY{\pA}{\pB}$. 
Then, there is an $\eps$-quadtree $\qte$ that has $\cell$ as a node, and let
$\cell_\pA$ and $\cell_\pB$ be the two children of $\cell$ in $\qte$,
containing $\pA$ and $\pB$ respectively. Furthermore, there is an
ordering $\order \in \OrdY{1/\eps}{d}$, such that $\cell_\pA$ and
$\cell_\pB$ are adjacent. As such, the cube $\cell_\pA$ (resp.,
$\cell_\pB$) corresponds to an interval $[x, x')$ (resp., $[x', x'')$)
in the ordering $(\qte,\order)$, and these two intervals are adjacent.
In particular, this implies that all points lying between $\pA$ and
$\pB$ in $\order$ have distance at most $2 \eps \distY{\pA}{\pB}$ from
$\pA$ or $\pB$.

If the above statement were true for all pairs of points, then this
would imply the main result (\thmref{lso}). However, consider the case
when there are two points close together, but no appropriately sized
quadtree cell contains both $\pt$ and $\ptq$.  In other words, two
points that are close together might get separated by nodes that are
much bigger in the quadtree, and this would not provide the guarantee
of the main result. However, this issue can be resolved using
shifting. We need the following result of Chan 
\cite[Lemma 3.3]{c-annqr-98}---a proof is provided in 
\apndref{shifting}.

\SaveContent{\LemmaShiftingWorks}%
{%
   Consider any two points $\pA, \pB \in [0,1)^d$, and let $\qt$ be
   the infinite quadtree of $[0,2)^d$.  For $D= 2\ceil{d/2}$ and
   $i = 0, \ldots, D$, let $v_i = (i/(D+1), \ldots, i/(D+1))$. Then
   there exists an $i \in \brc{0, \ldots, D}$, such that $\pA + v_i$
   and $\pB + v_i$ are contained in a cell of $\qt$ with side length
   $\leq 2(D+1) \distY{\pA}{\pB}$.%
}%

\begin{lemma}
    \lemlab{shifting:works}%
    \LemmaShiftingWorks
\end{lemma}

\subsection{Comparing two points according to an %
   \TPDF{$\eps$}{epsilon}-ordering}
\seclab{compare}

We now show how to efficiently compare two points in $\pset$ according
to a given $\eps$-ordering $\order$ with a shift $v_i$.  The shift can
be added to the two points directly, and as such we can focus on 
comparing two points according to $\order$.

First, we show how to compare the $\msb$ of two numbers using only
bitwise-exclusive-or and bitwise-and operations. We remark that
\obsref{msb:compare} \itemref{m:s:b} is from Chan \cite{c-cppsr-02}.

\begin{observation}
    \obslab{msb:compare}%
    \SICOMPVer{\sloppy}%
    Let $\xor$ denote the bitwise-exclusive-or operator. Let
    $\msb(a) := -\floor{\lg a}$ to be the index of the most
    significant bit in the binary expansion of $a \in [0,2)$.  Given
    $a, b \in [0,2)$, one can compare the $\msb$ of two numbers using
    the following:
    \begin{compactenumA}[leftmargin=0.4in]
        \NotSICOMPVer{\smallskip}%
        \item \itemlab{m:s:b} $\msb(a) > \msb(b)$ if and only if
        $a < b$ and $a < a \xor b$.

        \NotSICOMPVer{\smallskip}%
        \item \itemlab{m:s:b:e:q} $\msb(a) = \msb(b)$ if and only if 
        $a \xor b \leq a \land b$, where $\land$ is the bitwise-and 
        operator.
    \end{compactenumA}
\end{observation}
\begin{proof}
    (A) Observe that if $\msb(a) > \msb(b)$, then
    \begin{math}
        2^{-\msb(a)}%
        \leq%
        a%
        <%
        2^{-\msb(a)+1}%
        \leq%
        2^{-\msb(b)}%
        \leq%
        b.
    \end{math}
    Since $\msb(a) > \msb(b)$ and $a < b$, we have
    $\msb(a \oplus b) = \msb(b)$. As such, we have
    \begin{math}
        a%
        <%
        2^{-\msb(a)+1}%
        \leq%
        2^{-\msb(b)}%
        =%
        2^{-\msb(a \oplus b)}%
        \leq%
        a \oplus b.
    \end{math}
    
    Assume that $a < b$ and $a < a \xor b$.  Since $a < b$, it must be
    that $\msb(a) \geq \msb(b)$.  Observe that if $\msb(a) = \msb(b)$,
    then $a \xor b < a$, which is impossible. It follows that
    $\msb(a) > \msb(b)$, as desired.

    (B) Follows by applying \itemref{m:s:b} twice (in addition to
    using the inequalities $a \land b \leq a$ and $a \land b \leq b$),
    one can show that $a \xor b \leq a \land b$ if and only if
    $\msb(a) \geq \msb(b)$ and $\msb(b) \geq \msb(a)$.
\end{proof}

\SICOMPVer{\smallskip}%

\begin{lemma}
    \lemlab{compare:order}%
    Let $\pt = (\pt_1, \ldots, \pt_d)$ and
    $\ptq = (\ptq_1, \ldots, \ptq_d)$ be two distinct points in
    $\pset \subseteq [0,2)^d$ and $\order \in \orderset$ be an
    $\eps$-ordering over the cells of some $\eps$-quadtree $\qte$
    storing $\pset$. Then one can determine if $\pt \prec_\order \ptq$
    using $O(d\logeps)$ bitwise-logical operations.
\end{lemma}
\begin{proof}
    Recall $\eps$ is a power of two and $E = \lgeps$. In order to
    compare $\pt$ and $\ptq$, for $i = 1, \ldots, d$, compute
    $a_i = \pt_i \xor \ptq_i$. Find an index $i'$ such that
    $\msb(a_{i'}) \leq \msb(a_i)$ for all $i$. Such an index can be
    computed with $O(d)$ $\msb$ comparisons (using
    \obsref{msb:compare} \itemref{m:s:b}).  Given $\pt_{i'}$ and
    $\ptq_{i'}$, we next determine the place in which $\pt_{i'}$ and
    $\ptq_{i'}$ first differ in their binary representation. Note that
    because $\eps$ is a power of two, each digit in the base $1/\eps$
    expansion of $\pt_{i'}$ corresponds to a block of $E$ bits in the
    binary expansion of $\pt_{i'}$. Suppose that $\pt_{i'}$ and
    $\ptq_{i'}$ first differ inside the $h$\th block at an index
    $j \in \{1, \ldots E\}$.
    
    The algorithm now locates this index $j$. To do so, for 
    $j = 1, \ldots, E$, let  $b_j = 2^{E-j}/(2^E -1)\in (0,1]$ be the 
    number whose binary expansion has a 1 in positions 
    $j, j + E, j + 2E, \ldots$, and 0 everywhere else.  For 
    $j = 1, \ldots, E$, compute $b_j \land a_{i'}$ and check if 
    $\msb(a_{i'}) = \msb(b_j \land a_{i'})$ (using 
    \obsref{msb:compare} \itemref{m:s:b:e:q}). When the algorithm finds
    such an index $j$ obeying this equality, it exits the loop. We
    know that $\pt_{i'}$ and $\ptq_{i'}$ first differ in the $j$\th
    position inside the $h$\th block (the value of $h$ is never
    explicitly computed).

    It remains to extract the $E$ bits from the $h$\th block in each
    coordinate $\pt_1, \ldots, \pt_d$. For $i = 1, \ldots, d$, let
    $B_i \in \zeroone^E$ be the bits inside the $h$\th block of
    $\pt_i$. For $k = 1, \ldots, E$, set
    $B_{i,k} = \indic{\msb(2^{j-k} a_{i'}) = \msb((2^{j-k} a_{i'})
       \land \pt_i)}$ (where $\indic{\cdot}$ is the indicator
    function). By repeating a similar process for all
    $\ptq_1, \ldots, \ptq_d$, we obtain the coordinates of the cells
    in which $\pt$ and $\ptq$ differ. We can then consult $\order$ to
    determine whether or not $\pt \prec_\order \ptq$.

    This implies that $\pt$ and $\ptq$ can be compared using
    $O(d\logeps)$ operations by \obsref{msb:compare}.
\end{proof}

\paragraph*{Remark}
In the word RAM model for integer input, the extra $\log(1/\eps)$
factor in the above time bound can be eliminated: $\msb$ can be
explicitly computed in $O(1)$ time by a complicated algorithm of
Fredman and Willard~\cite{fw-sitbf-93}; this allows us to directly
jump to the right block of each coordinate and extract the relevant
bits. (Furthermore, assembly operations performing such computations
are nowadays available on most CPU{}s.)

\subsection{The result}

\begin{theorem}
    \thmlab{lso}%
    For $\epsA \in (0,1/2]$, there is a set $\ordAll$ of
    $\Od(\log (1/\epsA)/\epsA^d)$ orderings of $[0,1)^d$, such
    that for any two points $\pA, \pB \in [0,1)^d$ there is an
    ordering $\order \in \ordAll$ defined over $[0,1)^d$, such that
    for any point $\pC$ with $\pA \prec_\order \pC \prec_\order \pB$
    it holds that either
    $\distY{\pA}{\pC} \leq \epsA \distY{\pA}{\pB}$ or
    $\distY{\pB}{\pC} \leq \epsA \distY{\pA}{\pB}$.

    Furthermore, given such an ordering $\order$, and two points
    $\pA,\pB$, one can compute their ordering, according to $\order$,
    using $O(d\logeps)$ arithmetic and bitwise-logical operations.
\end{theorem}

\begin{proof}

    Let $\ordAll$ be the set of all orderings defined by picking an
    ordering from $\orderset$, as defined by \defref{o:set} using the
    parameter $\eps$, together with a shift from
    \lemref{shifting:works}.

    Consider any two points $\pA,\pB \in [0,1)^d$. By
    \lemref{shifting:works} there is a shift for which the two points
    fall into a quadtree cell $\cell$ with side length at most
    $2(D+1) \distY{\pA}{\pB}$. Next, there is an $\eps$-quadtree
    $\qte$ that contains $\cell$, and the two children that correspond
    to two cells $\cell_\pA$ and $\cell_\pB$ with side length at most
    $2(D+1)\epsA\distY{\pA}{\pB}$, which readily implies that the
    diameter of these cells is at most
    $2(D+1)\sqrt{d}\epsA\distY{\pA}{\pB}$. Furthermore, there is an
    $\eps$-ordering in $\orderset$ such that all the points of
    $\cell_\pA$ are adjacent to all the points of $\cell_\pB$ in this
    ordering.  This implies the desired claim, after adjusting $\eps$
    by a factor of $2(D+1)\sqrt{d}$ (and rounding to a power of 2).
\end{proof}

From now on, we refer to the set of orderings $\ordAll$ in the above
Theorem as \lso{}s.  We remark that by the readjustment of $\eps$ in
the final step of the proof, the number of \lso{}s when including the
factors involving $d$ is $O(d^{3/2})^d\cdot (1/\eps^d)\logeps$.

\subsubsection{Discussion}

\paragraph*{Connection to locality-sensitive hashing}

Let $\pset$ be a set of $n$ points in Hamming space $\brc{0,1}^d$.
Consider the decision version of the $(1 + \eps)$-approximate nearest
neighbor problem. Specifically, for a pre-specified radius $r$ and any
given query point $\ptq$, we would like to efficiently decide whether
or not there exists a point $\pt \in \pset$ such that
$\distY{\ptq}{\pt}_1 \leq (1 + \eps)r$ or conclude that all points in
$\pset$ are at least distance $r$ from $\ptq$.  The locality-sensitive
hashing (\LSH) technique \cite{im-anntrcd-98} implies the existence of
a data structure supporting this type of decision query in time
$O(dn^{1/(1 + \eps)} \log n)$ time (which is correct with high
probability) and using total space $O(dn^{1 +1/(1+\eps)}\log n)$.
Similar results also hold in the Euclidean setting.

At a high level, \LSH works as follows. Start by choosing 
$k := k(\eps, r, n)$ indices in $[d]$ at random (with replacement). 
Let $R$ denote the resulting multiset of coordinates. For each point 
$\pt \in \pset$, let $\pt_R$ be the projection $\pt$ onto these 
coordinates of $R$. We can group the points of $\pset$ into buckets, 
where each bucket contains points with the same projection. Given a 
query point $\ptq$, we check if any of the points in the same bucket
as $\ptq$ is at distance at most $(1 + \eps)r$ from $\ptq$. This 
construction can also be repeated a sufficient number of times in 
order to guarantee success with high probability.

The idea of bucketing can also be viewed as an implicit ordering on
the randomly projected point set by ordering points lexicographically
according to the $k$ coordinates. In this sense, the query algorithm
can be viewed as locating $\ptq$ within each of the orderings, and
comparing $\ptq$ to similar points nearby in each ordering. From this
perspective, every \lso{} can be viewed as an \LSH scheme. Indeed, 
for a given query point $\ptq$, the approximate nearest neighbor to
$\ptq$ can be found by inspecting the elements adjacent to $\ptq$ in
each of the \lso{}s and returning the closest point to $\ptq$ found
(see \thmref{dynamic:nn}).

Of course, the main difference between the two schemes is that for
every fixed $\eps$, the number of ``orderings'' in an \LSH scheme is
polynomial in both $d$ and $n$.  While for \lso{}s, the number of
orderings remains exponential in $d$.  This trade-off is to be expected, 
as \lso{}s guarantee a much stronger property than that of an \LSH 
scheme.

\paragraph*{Extension of \lso{}s to other norms in Euclidean space}
The $L_p$-norm, for $p \geq 1$, of a vector $x \in \Re^d$ is defined as
\begin{math}
    \| x \|_p = \pth{\abs{x_1}^p + \cdots + \abs{x_n}^p}^{1/p}.
\end{math}
The $L_\infty$-norm, or maximum norm, is defined as
\begin{math}
    \| x \|_\infty = \max\pth{\abs{x_1}, \ldots, \abs{x_n}}.
\end{math}

The result of \thmref{lso} also holds for any $L_p$-norm. 
The key change that is needed is in the proof of \lemref{shifting:works}:
For any two points $\pts, \ptt \in [0,1)^d$, there exists
a shift $v$ such that $\pts + v$ and $\ptt + v$ are contained
in a quadtree cell of side length at most $2(D+1)\distY{\pts}{\ptt}_p$.
This extension follows easily from the proof of the Lemma (see
\apndref{shifting}). \thmref{lso} then follows by adjusting
$\eps$ by a factor of $2(D+1)d^{1/p}$ in the last step,
implying that the number of orderings will be 
$O(d^{1 + 1/p})^d (1/\eps^d) \logeps$. (For the $L_\infty$-norm, 
$\eps$ only needs to be adjusted by a factor of $2(D+1)$.)

\paragraph*{Extension of \lso{}s for doubling metrics}

An abstraction of low-dimensional Euclidean space, is a metric space
with (low) doubling dimension. Formally, a metric space
$(\MetricSpace, \MetricDist)$ has \emph{doubling dimension} $\lambda$
if any ball of $\MetricSpace$ of radius $r$ can be covered by at most
$2^\lambda$ balls of half the radius (i.e., $r/2$). It is known that
$\Re^d$ has doubling dimension $O(d)$ \cite{v-cbseb-05}.  We point
out that locality-sensitive orderings still exist in this case, but
they are less constructive in nature, since one needs to be provided
with all the points of interest in advance.

For a point set $\PS \subseteq \MetricSpace$, the analogue of a
quadtree for a metric space is a net tree \cite{hm-fcnld-06}.  A net
tree can be constructed as follows (the construction algorithm
described here is somewhat imprecise): The root node corresponds to
the point set $\PS \subseteq \MetricSpace$. Compute a randomized
partition of $\PS$ of diameter 1/2 (assume $\PS$ has diameter one),
and for each cluster in the partition, create an associated node and
hang it on the root.  The tree is computed recursively in this manner,
at each level $i$ computing a random partition of diameter
$2^{-i}$. The leaves of the tree are points of $\PS$.

As with quadtrees, it is possible during this randomized construction
for two nearby points to be placed in different clusters and be
separated further down the tree. If $\ell = \MetricDist(\pt,\ptq)$ for
two points $\pt, \ptq \in \PS$, then the probability that $\pt$ and
$\ptq$ lie in different clusters of diameter $r = 2^{-i}$ in the
randomized partition is at most $O((\ell/r) \log n)$
\cite{frt-tbaam-04}. In particular, for $r \approx 1/(\ell \log n)$,
the probability that $\pt$ and $\ptq$ are separated is at most a
constant. If we want this property to hold with high probability
for all pairs of points, one needs to construct $O(\log n)$ (randomly
constructed) net trees of $\PS$. (This corresponds to randomly
shifting a quadtree $O(\log n)$ times in the Euclidean setting.)

Given such a net tree $T$, each node has $I = 2^{O(\lambda)}$
children. We can arbitrarily and explicitly number the children of
each node by a distinct label from $\IRX{I}$. One can define an
ordering of such a tree as we did in the Euclidean case, except that
the gap (in diameter) between a node and its children is
$O(\eps/\log n)$ instead of $\eps$. Repeating our scheme in the
Euclidean case, this implies that one would expect to require
$ (\eps^{-1}\log n)^{O(\lambda)}$ orderings of $\PS$.

This requires having all the points of $\PS$ in advance, which is a
strong assumption for a dynamic data structure (as in some of the
applications below). For example, Gottlieb and Roditty
\cite{gr-odsdm-08} show how to maintain dynamic spanners in a doubling
metric, but only assuming that after a point has been deleted from
$\PS$, the distance between the deleted point and a point currently in
$\PS$ can still be computed in constant time.

\section{Applications}

\subsection{Bichromatic closest pair}

Given an ordering $\order \in \ordAll$, and two finite sets of points
$\rset, \bset $ in $\Re^d$, let
$\AccSet = \AccSet(\order, \rset,\bset)$ be the set of all pairs of
points in $\rset \times \bset$ that are adjacent in the ordering of
$\rset \cup \bset$ according to $\order$.  Observe that inserting or
deleting a single point from these two sets changes the contents of
$\AccSet$ by a constant number of pairs.  Furthermore, a point
participates in at most two pairs.

\begin{lemma}
    \lemlab{1d:closest:pair}%
    Let $\rset$ and $\bset$ be two sets of points in $[0,1)^d$, and
    let $\eps \in(0,1)$ be a parameter. Let $\order \in \ordAll$ be a
    \lso (see \thmref{lso}). Then, one can maintain the set
    $\AccSet = \AccSet(\order, \rset,\bset)$ under insertions and
    deletions to $\rset$ and $\bset$. In addition, one can maintain
    the closest pair in $\AccSet$ (under the Euclidean metric). Each
    update takes $O(d\log n\logeps)$ time, where $n$ is the total
    size of $\rset$ and $\bset$ during the update operation.
\end{lemma}
\begin{proof}
    Maintain two balanced binary search trees $\Tree_\rset$ and
    $\Tree_\bset$ storing the points in $\rset$ and $\bset$,
    respectively, according to the order $\order$.  Insertion,
    deletion, predecessor query and successor query can be implemented
    in $O(d \logeps \log n)$ time (since any query requires
    $O(\log n)$ comparisons each costing $O(d\logeps)$ time by
    \lemref{compare:order}). We also maintain a min-heap of the 
    pairs in $\AccSet$ sorted according to the Euclidean
    distance. The minimum is the desired closest pair.
    Notice that a single point can participate in at most two
    pairs in $\AccSet$.

    We now explain how to handle updates.
    Given a newly inserted point $\ptr$ (say a red point that
    belongs to $\rset$), we compute the (potential) pairs it
    participates in, by computing its successor $\ptr'$ in $\rset$, and
    its successor $\ptb'$ in $\bset$. If
    $\ptr \prec_\order \ptb' \prec_\order \ptr'$ then the new pair
    $\ptr \ptb'$ should be added to $\AccSet$. The pair before $\ptr$
    in the ordering that might use $\ptr$ is computed in a similar
    fashion.  In addition, we recompute the predecessor and successor
    of $\ptr$ in $\rset$, and we recompute the pairs they might
    participate in (deleting potentially old pairs that are no longer
    valid).

    Deletion is handled in a similar fashion---all points included in
    pairs with the deleted point recompute their pairs. In addition,
    the successor and predecessor (of the same color) need to
    recompute their pairs. This all requires a constant number of
    queries in the two trees, and thus takes the running time as
    stated.
\end{proof}

\begin{theorem}
    \thmlab{bichromatic:closest:pair}%
    Let $\rset$ and $\bset$ be two sets of points in $[0,1)^d$, and
    let $\eps \in (0,1/2]$ be a parameter. Then one can maintain a
    $(1+\eps)$-approximation to the bichromatic closest pair in
    $\rset \times \bset$ under updates (i.e., insertions and
    deletions) in $\Od(\log n \log^2(1/\eps)/\eps^d)$ time per
    operation, where $n$ is the total number of points in the two
    sets. The data structure uses $\Od(n\logeps/\eps^d)$ space, and
    at all times maintains a pair of points 
    $\ptr \in \rset,\ \ptb \in \bset$, such that 
    $\distY{\ptr}{\ptb} \leq (1+\eps) \distSetY{\rset}{\bset}$,
    where
    $\distSetY{\rset}{\bset} = \min_{\ptr \in \rset, \ptb \in \bset}
    \distY{\ptr}{\ptb}$.
\end{theorem}
\begin{proof}
    We maintain the data structure of \lemref{1d:closest:pair} for all
    the \lso{}s of \thmref{lso}. %
    All the good pairs for these data structures can be maintained
    together in one global min-heap. The claim is that the minimum
    length pair in this heap is the desired approximation.

    To see that, consider the bichromatic closest pair
    $\ptr \in \rset$ and $\ptb \in \bset$.  By \thmref{lso} there is a
    \lso $\order$, such that the interval $I$ in the ordering between
    $\ptr$ and $\ptb$ contains points that are in distance at most
    $\ell = \epsA \distY{\ptr}{\ptb}$ from either $\ptr$ or $\ptb$. In
    particular, let $\PS_\ptr$ (resp., $\PS_\ptb$) be all the points
    in $I$ in distance at most $\ell$ from $\ptr$ (resp.,
    $\ptb$). Observe that $\PS_\ptr \subseteq \rset$, as otherwise,
    there would be a bichromatic pair in $\PS_\rset$, and since the
    diameter of this set is at most $\ell$, this would imply that
    $(\ptr,\ptb)$ is not the closest bichromatic pair---a
    contradiction. Similarly, $\PS_\ptb \subseteq \bset$. As such,
    there must be two points $\ptb' \in \bset$ and $\ptr' \in \rset$,
    that are consecutive in $\order$, and this is one of the pairs
    considered by the algorithm (as it is stored in the min-heap). In
    particular, by the triangle inequality, we have
    \begin{equation*}
        \distY{\ptr'}{\ptb'}%
        \leq%
        \distY{\ptr'}{\ptr} +\distY{\ptr}{\ptb} +\distY{\ptb}{\ptb'}
        \leq%
        2 \ell + \distY{\ptr}{\ptb}%
        \leq%
        (1+2\eps)\distY{\ptr}{\ptb}.
    \end{equation*}
    The theorem follows after adjusting $\eps$ by a factor of 2.
\end{proof}

\paragraph*{Remark}
In the word RAM model, for integer input in $\{1,\ldots,U\}^d$, the
update time can be improved to $\Od((\log\log U)\log^2(1/\eps)/\eps^d)$
expected, by using van Emde Boas trees~\cite{e-poflt-77} in place of
the binary search trees (and the min-heaps as well).  With standard
word operations, we may not be able to explicitly map each point to an
integer in one dimension following each locality-sensitive ordering,
but we can still simulate van Emde Boas trees on the input as if the
mapping has been applied.  Each recursive call in the van Emde Boas
recursion focuses on a specific block of bits of each input coordinate
value (after shifting); we can extract these blocks, and perform the
needed hashing operations on the concatenation of these blocks over
the $d$ coordinates of each point.

\subsection{Dynamic spanners}
\seclab{spanners}
\begin{defn}
\deflab{spanner} 
For a set of $n$ points $\pset$ in $\Re^d$ and a
parameter $t \geq 1$, a \emphi{$t$-spanner} of $\pset$ is an
undirected graph $\graph = (\pset, \graphedges)$ such that for all
$\pt, \ptq \in \pset$,
$$
\dist{\pt - \ptq} \leq \graphdist{\pt}{\ptq} \leq t \dist{\pt - \ptq},
$$
where $\graphdist{\pt}{\ptq}$ is the length of the shortest path from
$\pt$ to $\ptq$ in $\graph$ using the edge set $\graphedges$.
\end{defn}

Using a small modification of the results in the previous section, we
easily obtain a dynamic $(1+\eps)$-spanner. Note that there is nothing
special about how the data structure in
\thmref{bichromatic:closest:pair} deals with the bichromatic point
set. If the point set is monochromatic, modifying the data structure
in \lemref{1d:closest:pair} to account for the closest monochromatic
pair of points leads to a data structure with the same bounds and
maintains the $(1 + \eps)$-approximate closest pair.

The construction of the spanner is very simple: Given $\pset$ and
$\epsA \in (0,1)$, maintain orderings of the points specified by
$\ordAll$ (see \thmref{lso}). For each $\order \in \ordAll$, let
$\graphedges_\order$ be the edge set consisting of edges connecting
two consecutive points according to $\sigma$,
with weight equal to their Euclidean distance. Thus
$\cardin{\graphedges_\order} = n-1$. Our spanner
$\graph = (\pset, \graphedges)$ then consists of the edge set
$\graphedges = \bigcup_{\order \in \ordAll} \graphedges_\order$.

\begin{theorem}
    \thmlab{dynamic:spanners}%
    Let $\pset$ be a set of $n$ points in $[0,1)^d$ and
    $\eps \in (0,1/2]$. One can compute a $(1 + \eps)$-spanner $\graph$
    of $\pset$ with $\Od(n\logeps/\eps^d)$ edges, where every vertex has
    degree $\Od(\logeps/\eps^d)$. Furthermore, a point can be inserted or
    deleted in $\Od(\log n \log^2(1/\eps)/\eps^d)$ time, where each
    insertion or deletion creates or removes at most 
    $\Od(\logeps/\eps^d)$ edges in the spanner.
\end{theorem}
\begin{proof}
    The construction is described above. The same analysis as in the
    proof of \thmref{bichromatic:closest:pair} implies the number of
    edges in $\graph$ and the update time.
    
    It remains to prove that $\graph$ is a spanner.  By \thmref{lso},
    for any pair of points $\pts, \ptt \in \pset$, there is a \lso
    $\order \in \ordAll$, such that the $\order$-interval
    $[\pts,\ptt)$ contains only points that are in distance at most
    $\epsA \distY{\pts}{\ptt}$ from either $\pts$ or $\ptt$. In
    particular, there must be two points in $\pts',\ptt' \in \pset$
    that are adjacent in $\order$, such that one of them, say $\pts'$
    (resp., $\ptt'$) is in distance at most $\epsA \distY{\pts}{\ptt}$
    from $\pts$ (resp., $\ptt$). As such, the edge $\pts'\ptt'$ exists
    in the graph being maintained.

    This property is already enough to imply that this graph is a
    $(1+c\eps)$-spanner for a sufficiently large constant~$c$---this
    follows by an induction on the distances between the points 
    (specifically, in the above, we apply the induction hypothesis
    on the pairs $s,s'$ and $t,t'$). 
    We omit the easy but somewhat tedious argument---see
    \cite{ck-fasgg-93} or \cite[Theorem 3.12]{h-gaa-11} for details.
    The theorem follows after adjusting $\eps$ by a factor of~$c$.
\end{proof}

\subsubsection{Static and dynamic vertex-fault-tolerant spanners}
\begin{defn}
    \deflab{fault:tolerant:spanner}%
    For a set of $n$ points $\pset$ in $\Re^d$ and a parameter
    $t \geq 1$, a \emphi{$k$-vertex-fault-tolerant $t$-spanner} of
    $\pset$, denoted by $(k,t)$-\VFTS, is a graph
    $\graph = (\pset, \graphedges)$ such that
    \begin{compactenumi}
        \item $\graph$ is a $t$-spanner (see \defref{spanner}), and
        \item For any $\pset' \subseteq \pset$ of size at most $k$,
        the graph $\graph \setminus \pset'$ is a $t$-spanner for
        $\pset \setminus \pset'$.
    \end{compactenumi}
\end{defn}

A $(k,1 + \eps)$-\VFTS can be obtained by modifying the construction
of the $(1 + \eps)$-spanner in \secref{spanners}.
Construct a set of \lso{}s $\ordAll$.
For each $\order \in \ordAll$ and each $\pt \in \pset$, connect $\pt$
to its $k+1$ successors and $k+1$ predecessors according to $\order$
with edge weights equal to the Euclidean distances. Thus each ordering
maintains $O(nk )$ edges and there are
$O(\cardin{\ordAll} kn ) = \Od(kn \log(1/\eps)/\eps^d)$ edges
overall. We now prove that this graph $\graph$ is in fact a
$(k,1 + \eps)$-\VFTS.

\begin{theorem}
    \thmlab{ft:spanners}%
    Let $\pset$ be a set of $n$ points in $[0,1)^d$ and
    $\eps \in (0,1/2]$.  One can compute a $k$-vertex-fault-tolerant
    $(1 + \eps)$-spanner $\graph$ for $\pset$ in time
    \NotSICOMPVer{
    $\Od\pth{(n\log n\log(1/\eps) + kn)\log(1/\eps)/\eps^d}$.
    }
    \SICOMPVer{
    \begin{align*}
    \Od\pth{(n\log n\log(1/\eps) + kn)\log(1/\eps)/\eps^d}.  
    \end{align*}
    }
    The number of edges is $\Od(kn\log(1/\eps)/\eps^d )$ and the 
    maximum degree is bounded by $\Od(k\log(1/\eps)/\eps^d)$.

    Furthermore, one can maintain the $k$-vertex-fault-tolerant
    $(1 + \eps)$-spanner $\graph$ under insertions and deletions of 
    points in
    $\Od\bigl((\log n\log(1/\eps) + k) \log(1/\eps)/\eps^d\bigr)$ 
    time per operation.
\end{theorem}
\begin{proof}
    The construction algorithm, number of edges, and maximum degree
    follows from the discussion above. So, consider deleting a set
    $\pset' \subseteq \pset$ of size at most $k$ from
    $\graph$. Consider an ordering $\order \in \ordAll$ with the
    points $\pset'$ removed. By the construction of $\graph$, all the
    pairs of points of $\pset \setminus \pset'$ that are (now)
    adjacent in $\order$ remain connected by an edge in
    $\graph \setminus \pset'$. The argument of
    \thmref{dynamic:spanners} implies that the remaining graph is
    spanner. We conclude that $\graph \setminus \pset'$ is a
    $(1 + \eps)$-spanner for $\pset \setminus \pset'$.

    As for the time taken to handle insertions and deletions, one
    simply maintains the orderings of the points using balanced search
    trees. After an insertion of a point to one of the orderings in
    $O(\log n \logeps )$ time, $O(k)$ edges have to be added and
    deleted. Therefore inserting a point takes
    $O\bigl((\log n \logeps + k) \cardin{\ordAll}\bigr)
    = \Od\bigl((\log n\log(1/\eps) + k) \log(1/\eps)/\eps^d\bigr)$
    time total. Deletions are handled similarly.
    
    The total construction time follows by inserting each of 
    the points into the dynamic data structure.
\end{proof}

\subsection{Dynamic approximate nearest neighbors}
Another application of the same data structure in
\thmref{bichromatic:closest:pair} is supporting
$(1 + \eps)$-approximate nearest neighbor queries. In this scenario,
the data structure must support insertions and deletions of points and
the following queries: given a point $\ptq$, return a point
$\ptt \in \pset$ such that
$\distY{\ptq}{\ptt} \leq (1 + \eps)\min_{p\in P}\distY{\ptq}{p}$.

\begin{theorem}
    \thmlab{dynamic:nn}%
    \SICOMPVer{\sloppy}
    Let $\pset$ be a set of $n$ points in $[0,1)^d$. For a given
    $\eps \in (0,1/2]$, one can build a data structure using
    $\Od(n\logeps/\eps^d)$ space, that supports insertion and deletion
    in time $\Od(\log n \log^2(1/\eps)/\eps^d)$. Furthermore, given a
    query point $\ptq \in [0,1)^d$, the data structure returns a
    $(1+\eps)$-approximate nearest neighbor in $\pset$ in
    $\Od(\log n\log^2(1/\eps)/\eps^{d})$ time.
\end{theorem}
\begin{proof}
    Maintain the data structure of \lemref{1d:closest:pair} for all
    \lso{}s of \thmref{lso}, with one difference: Since the input is
    monochromatic, for each \lso{} $\order \in \ordAll$, we store the
    points in a balanced binary search tree according to $\order$. The 
    space and update time bounds easily follow by the same analysis.

    Given a query point $\ptq \in [0,1)^d$, for each of the orderings
    the algorithm inspects the predecessor and successor to
    $\ptq$. The algorithm returns the closest point to $\ptq$
    encountered.  We claim that the returned point $\pt$ is the
    desired approximate nearest neighbor.

    Let $\ptopt \in \pset$ be the nearest neighbor to $\ptq$ and
    $\ell = \dist{\ptq - \ptopt}$. By \thmref{lso}, there is a \lso{}
    $\order \in \ordAll$ such that the $\order$-interval
    $I = [\ptopt, \ptq)$ contains points that are of distance at most
    $\eps \ell$ from $\ptopt$ or $\ptq$ (and this interval contains at
    least one point of $\pset$, namely, $\ptopt$). Note that no point
    of $\pset$ can be at distance less than $\eps \ell$ to
    $\ptq$. %
    Thus, the point $\pt \in \pset$ adjacent to $\ptq$ in $I$ is of
    distance at most $\eps\ell$ from $\ptopt$. Therefore, for such a
    point $\pt$, we have
    $\distY{\pt}{\ptq} \leq \distY{\pt}{\ptopt} + \distY{\ptopt}{\ptq}
    \leq (1 + \eps)\ell$.
    
    The final query time follows from the time taken for these
    predecessor and successor queries, as in the proof of
    \lemref{1d:closest:pair}.
\end{proof}

\section{Conclusion}

In this paper, we showed that any bounded subset of $\Re^d$ has a
collection of ``few'' orderings which captures proximity.  This
readily leads to simplified and improved approximate dynamic data
structures for many fundamental proximity-based problems in
computational geometry. Beyond these improvements, we believe that the
new technique could potentially be simple enough to be useful in
practice, and could be easily taught in an undergraduate level class
(replacing, for example, well-separated pair decomposition---a topic
that is not as easily accessible).

We expect other applications to follow from the technique presented in
this paper. For example, recently Buchin \etal~\cite{bho-spda-19}
presented a near linear-sized construction for robust spanners. The
idea is to build a robust spanner in one dimension, and then obtain a
robust spanner in higher dimensions by applying the one-dimensional
construction using the \lso{}s.

\paragraph*{Acknowledgments.}

The authors thank the anonymous referees for their detailed and useful
comments.

 \providecommand{\CNFX}[1]{ {\em{\textrm{(#1)}}}}

\appendix
\section{Proofs}

\subsection{Proof of \lemref{many:orderings}}
\apndlab{many:orderings}

\begin{figure}
    
    \NotSICOMPVer{
    \begin{minipage}{0.24\linewidth}
        \includegraphics[page=3]{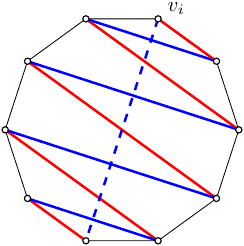}%
    \end{minipage}
    \hfill%
    \begin{minipage}{0.48\linewidth}
        \centerline{%
           \begin{tabular}{*{3}{c}}
             \includegraphics[page=4,scale=0.5]{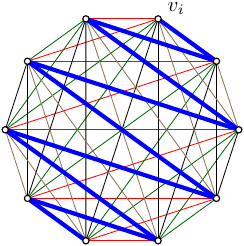}%
             &
               \includegraphics[page=5,scale=0.5]{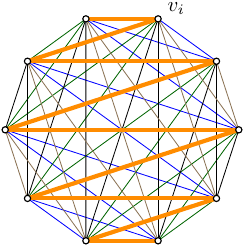}%
             &
               \includegraphics[page=6,scale=0.5]{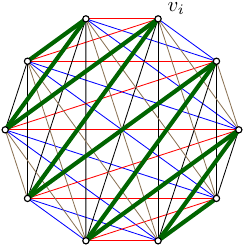}%
             \\
             \includegraphics[page=7,scale=0.5]{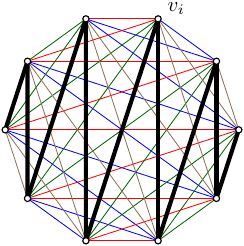}%
             &
               \includegraphics[page=8,scale=0.5]{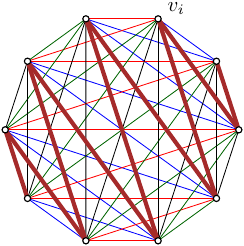}%
           \end{tabular}%
        }
    \end{minipage}
    \hfill
    \begin{minipage}{0.24\linewidth}
        \includegraphics[page=9]{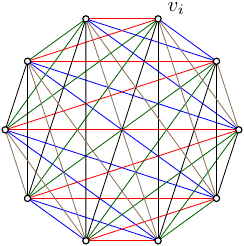}%
    \end{minipage}
    }
    
    \SICOMPVer{
    \begin{minipage}{0.24\linewidth}
    \includegraphics[page=3,scale=0.75]{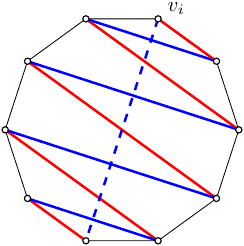}%
    \end{minipage}
    \hfill%
    \begin{minipage}{0.4\linewidth}
        \centerline{%
           \begin{tabular}{*{3}{c}}
             \includegraphics[page=4,scale=0.45]{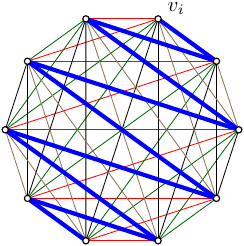}%
             &
               \includegraphics[page=5,scale=0.45]{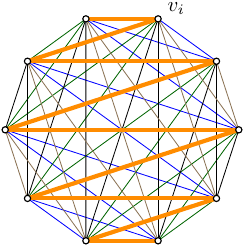}%
             &
               \includegraphics[page=6,scale=0.45]{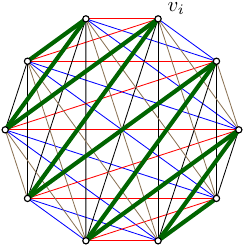}%
             \\
             \includegraphics[page=7,scale=0.45]{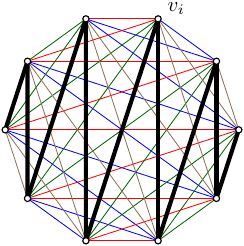}%
             &
               \includegraphics[page=8,scale=0.45]{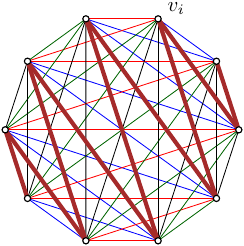}%
           \end{tabular}%
        }
    \end{minipage}
    \hfill
    \begin{minipage}{0.24\linewidth}
        \includegraphics[page=9,scale=0.75]{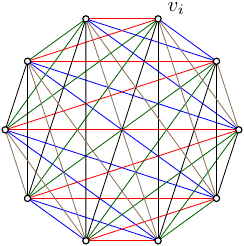}%
    \end{minipage}
    }
    
    \caption{For $n$ even, a decomposition of $K_n$ into $n/2$
       Hamiltonian paths.}
    \figlab{H:cover}%
\end{figure}

{\RestatementOf{\lemref{many:orderings}}{\LemmaManyOrderings}}

\begin{proof}
    As mentioned earlier this is well known \cite{a-wwc-08}.
    Assume $n$ is even, and consider the clique $K_{n}$, with its
    vertices $v_0,\ldots, v_{n-1}$. The edges of this clique can be
    covered by $n/2$ Hamiltonian paths that are edge disjoint. Tracing
    one of these path gives rise to one ordering, and doing this for
    all paths results with orderings with the desired property, since
    edge $v_iv_j$ is adjacent in one of these paths.
    
    To get this cover, draw $K_n$ by using the vertices of an
    $n$-regular polygon, and draw all the edges of $K_n$ as straight
    segments. For every edge $v_iv_{i+1}$ of $K_n$ there are exactly
    $n/2$ parallel edges with this slope (which form a matching). Let
    $M_i$ denote this matching. Similarly, for the vertex $v_i$,
    consider the segment $v_i v_{i+n/2}$ (indices are here modulo
    $n$), and the family of segments (i.e., edges) of $K_n$ that are
    orthogonal to this segment. This family is also a matching $M_i'$
    of size $n/2-1$. Observe that $\order_i = M_i \cup M_i'$ forms a
    Hamiltonian path, as shown in \figref{H:cover}. Since the slopes
    of the segments in $M_i$ and $M_i'$ are unique, for
    $i=0,\ldots, n/2-1$, it follows that
    $\order_0, \ldots, \order_{n/2-1}$ are an edge-disjoint cover of
    all the edges of $K_n$ by $n/2$ Hamiltonian paths.

    If $n$ is odd, use the above construction for $n+1$, and delete
    the redundant symbol from the computed orderings.
\end{proof}

\subsection{Proof of \lemref{shifting:works} (shifting)}
\apndlab{shifting}%

For two positive real numbers $x$ and $y$, let
\begin{equation*}
    x \Bmod y = x - y\floor{x/y}.   
\end{equation*}
The basic idea behind shifting is that one can pick a set of values
that look the ``same'' in all resolutions.

\begin{lemma}
    \lemlab{yey}%
    \SICOMPVer{\sloppy} Let $n>1$ be a positive odd integer, and
    consider the set
    \begin{equation*}
        X = \Set{ i/n}{i=0,\ldots, n-1 }.
    \end{equation*}
     Then, for any
    $\alpha = 2^{-\ell}$, where $\ell \geq 0$ is integer, we have that
    \begin{equation*}
        X \Bmod \alpha = \Set{ i/n \Bmod \alpha}{i=0,\ldots, n-1}        
    \end{equation*}
    is
    equal to the set
    \begin{math}
        \alpha X = \Set{ \alpha i/n}{i=0,\ldots, n-1}.
    \end{math}
\end{lemma}

\begin{proof}
    The proof is by induction. For $\ell =0$ the claim clearly
    holds. Next, assume the claim holds for some $i \geq 0$, and
    consider $\ell = i +1$. Setting $m = (n-1)/2$ and
    $\Delta = 2^{-i}/n$, we have by induction (and rearrangement) that
    \begin{align*}
        X \Bmod 2^{-i}%
        &=%
        2^{-i} X = \brc{ 0, \Delta, \ldots,
           2m\Delta}\\
        &=%
        \brc{ 0, (m+1) \Delta, \Delta, (m+2) \Delta,
           2\Delta, \ldots, (m+m) \Delta, m \Delta }.
    \end{align*}
    Setting $\delta = \Delta/2 = 2^{-i-1}/n$, we have
    \begin{align*}
      &X \Bmod 2^{-i-1}%
      =%
        \pth{X \Bmod  2^{-i}} \Bmod  2^{-i-1}%
      \\&
      =%
      \brc{ 0, (m+1) \Delta,  \Delta,
      (m+2) \Delta,  2\Delta, 
      \ldots,%
      (m+j) \Delta, \,\, j \Delta,\,
      \ldots,%
      (m+m) \Delta, m \Delta }\Bmod  2^{-i-1}
      \\&%
      =%
      \brc{ 0, 2(m+1) \delta,  2\delta,
      2(m+2) \delta,  4\delta,
      \ldots,
      2(m+j) \delta,\, 2j \delta,
      \ldots,
      2(m+m) \delta, 2m \delta }\Bmod  2^{-i-1}        
      \\&%
      =%
      \brc{ 0, \qquad \qquad\delta,  2\delta,
      \quad \quad\quad\; 3 \delta,  4\delta,
      \ldots,
      \,\,(2j-1) \delta, \,2j \delta,\,
      \ldots,
      (2m-1) \delta, 2m \delta },
    \end{align*}
    since $(2m+1)\delta = n \delta = 2^{-i-1}$ and
    $2(m+j )\delta \Bmod 2^{-i-1} = ( 2m+1 + 2j -1 )\delta \Bmod
    2^{-i-1} = (2j-1) \delta$, for $j = 1, \ldots, m$.
\end{proof}

\bigskip%

\RestatementOf{\lemref{shifting:works}}{\LemmaShiftingWorks}

\medskip%

\begin{proof}
    We start with the assumption that $d$ is even (this assumption
    will be removed at the end of the proof).
    Let $\ell \in \Na$, such that for $\alpha = 2^{-\ell}$, we
    have
    \begin{equation*}
        (d+1) \distY{\pA}{\pB}%
        <%
        \alpha%
        \leq
        2(d+1)\distY{\pA}{\pB}.
    \end{equation*}
    For $\tau \in [0,1]$, let $\grid + \tau$ denote the (infinite)
    grid with side length $\alpha$ shifted by the point
    $(\tau, \ldots, \tau)$.

    Let $X = \Set{ i/(d+1) }{i=0,\ldots, d}$ be the set of shifts
    considered. Since we are shifting a grid with side length $\alpha$,
    the shifting is periodical with value $\alpha$. It is thus
    sufficient to consider the shifts modulo $\alpha$.
        
    Let $\pA = (\pA_1,\ldots, \pA_d)$ and
    $\pB= (\pB_1,\ldots, \pB_d)$.  Assume that $\pA_1 \leq \pB_1$. A
    shift $\tau$ is bad, for the first coordinate, if there is an
    integer $i$, such that $\pA_1 \leq \tau + i \alpha \leq \pB_1$.
    The set of bad shifts in the interval $[0,\alpha]$ is
    \begin{equation*}
        B_1 = \Set{\bigl. (\pA_1,\pB_1) + i \alpha}{i \in \ZZ} \cap [0,\alpha].
    \end{equation*}
    The set $B_1$ is either an interval of length
    $|\pA_1 - \pB_1| \leq \distY{\pA}{\pB} < \alpha/(d+1)$, or two
    intervals (of the same total length) adjacent to $0$ and
    $\alpha$. In either case, $B_1$ can contain at most one point of
    $\alpha X = X \Bmod \alpha$, since the distance between any two
    values of $\alpha X$ is at least $\alpha/(d+1)$, by \lemref{yey}.

    Thus, the first coordinate rules out at most one candidate shift
    in $X \Bmod \alpha$.  Repeating the above argument for all $d$
    coordinates, we conclude that there is at least one shift in
    $\alpha X$ that is good for all coordinates. Let
    $\beta = \alpha i/(d+1) \in \alpha X$ this be good shift.  Namely,
    $\pA$ and $\pB$ belong to the same cell of $\grid + \beta$. The
    final step is to observe that shifting the points by $-\beta$,
    instead of the grid by distance $\beta$ has the same effect (and
    $-\beta \Bmod \alpha \in \alpha X$), and as such, the canonical
    cell containing both $\pA$ and $\pB$ is in the quadtree $\qt$ as
    desired, and the side length of this cell is $\alpha$.
    
    Finally, if $d$ is odd, replace $d$ by $d+1$ in the above proof. 
    This results in a set of $d + 2 = 2\ceil{d/2} + 1 = D + 1$ 
    shifts.
\end{proof}

\end{document}